\documentclass[11pt,leqno]{amsart}
\usepackage{amsmath,amsthm,amssymb,mathtools}
\usepackage{tabularx}
\usepackage{cleveref}
\usepackage{comment}
\usepackage{url} 
\usepackage{tikz}

\newcommand{\FF}{{\mathbb F}}

\newcommand{\ZZ}{{\mathbb Z}}

\newcommand{\RR}{{\mathbb R}}

\newcommand{\xx}{\mathbf{x}}
\newcommand{\yy}{\mathbf{y}}
\newcommand{\cc}{\mathbf{c}}
\newcommand{\uu}{\mathbf{u}}

\newcommand{\vv}{\mathbf{v}}

\newcommand{\itemc}{\item[$\cdot$]}

\usepackage{color} 
\usepackage{ulem}

\newcommand{\rev}[1]{#1}
\newcommand{\com}[2]{}

\newenvironment{pproof}[1]{\vspace{0.5\baselineskip}\noindent\textbf{Proof of #1.} }{\hfill$\square$\par\vspace{0.5\baselineskip}}

\newtheorem{thm}{Theorem}[section]
\newtheorem{lem}[thm]{Lemma}
\newtheorem{cor}[thm]{Corollary}
\newtheorem{prop}[thm]{Proposition}

\theoremstyle{definition}
\newtheorem{df}[thm]{Definition}
\newtheorem{rem}[thm]{Remark}

\numberwithin{equation}{section}
\title[On Lattice Isomorphism Problems, LCD Codes over Finite Rings]
{On Lattice Isomorphism Problems for Lattices from LCD Codes over Finite Rings}

\author[Nishimura]{Yusaku Nishimura*}
\thanks{*Corresponding author}
\address{School of Fundamental Science and Engineering, 
Waseda University, 
Tokyo 169--8555, Japan
}
\email{n2357y@ruri.waseda.jp} 

\author[Takashima]{Katsuyuki Takashima}
\address{	Faculty of Education and Integrated Arts and Sciences, School of Education, Waseda University
}
\email{ktakashima@waseda.jp}

\author[Miezaki]{Tsuyoshi Miezaki}
\address{		Faculty of Science and Engineering, 
		Waseda University, 
		Tokyo 169--8555, Japan
}
\email{miezaki@waseda.jp} 



\begin{document}
\begin{abstract}
  These days, post-quantum cryptography based on the lattice isomorphism problem has been proposed.
  Ducas-Gibbons introduced the hull attack, which solves the lattice isomorphism problem for lattices obtained by Construction A from an LCD code over a finite field.
  Using this attack, they showed that the lattice isomorphism problem for such lattices can be reduced to the lattice isomorphism problem with the trivial lattice $\mathbb{Z}^n$ and the graph isomorphism problem.
  While the previous work by Ducas-Gibbons only considered lattices constructed by a code over a \textit{finite field}, 
  this paper considers lattices constructed by a code over a \textit{finite ring} $\mathbb{Z}/k\mathbb{Z}$, which is a more general case.
  In particular, when $k$ is odd, an odd prime power, or not divisible by $4$, we show that the lattice isomorphism problem can be reduced to the lattice isomorphism problem for $\mathbb{Z}^n$ and the graph isomorphism problem.
\end{abstract}
\maketitle

%%%%%%%%%%%%%%%%%%%%%%%%%%%%%%%%%%%%%%%%%%%%%%%%%%%%%%%%%%%%%%%%%%%%%%%%%%%%%%

\section{Introduction}

\subsection{Background}

Lattice-based cryptosystems that are based on the lattice isomorphism problem (LIP), such as Hawk, are expected to be post-quantum secure~\cite{DPPW2022}.
Hawk is a round 2 candidate for additional signature NIST PQC competition, and it generates and verifies signatures efficiently on all devices, including low-end ones~\cite{Hawk}.
Due to these features, Hawk has attracted attention, and therefore, it is important to analyze the hardness of LIP.
Ducas and Gibbons~\cite{DG2023} introduce the hull attack, which solves LIP for lattices obtained by Construction A of LCD codes over finite fields,  
where LCD codes are codes that have no intersection with their dual codes.  
This property plays an important role in the hull attack, and they show that the LIP for such lattices are reduced to the LIP with the trivial lattice $\mathbb{Z}^n$ (called $\mathbb{Z}$LIP) and the graph isomorphism problem. %~\cite{DG2023}.
However, the hull attack can only be applied to lattices constructed from codes over finite fields, not to lattices generated by codes over finite rings, such as $\mathbb{Z}/k\mathbb{Z}$.
In this paper, we consider LIP constructed from codes not over finite fields, but over finite rings $\ZZ/k\ZZ$.
% 耐量子計算機暗号として、Hawkといった格子同型問題(LIP)に基づく暗号方式が提案されており~\cite{DPPW2022}、
% 格子同型問題の計算困難性の考察は重要な問題である。
% Ducas-Gibbonsによって、有限体上のLCD符号から構成法Aで得られる格子に対し、hull攻撃と呼ばれる格子同型問題解法が考案され~\cite{DG2023}、
% このような格子達の格子同型問題が$\ZZ^n$格子の同型問題($\ZZ$LIP)およびグラフ同型問題へと帰着できることが示された。
% この研究では位数が素数の有限体の符号から得られる格子のみの結果であり、例えば一般の整数$k$に対して、$\ZZ/k\ZZ$上の符号から得られる格子については言及されていない。
% そこで本稿では、有限体以外の符号から得られる格子についての格子同型問題解法を考察する。

% $L_1,L_2$の両方の最短ベクトルの集合がわかっている時、LIPはグラフ同型問題へと帰着できることが知られており~\cite{SHVvW2020}、
% 従ってLIPの困難性は、
% % その格子の幾何的性質にも依存するが、
% 基本的には最短ベクトル問題と同程度、あるいはそれ以上に難しい問題である。

\subsection{Previous work}

Ducas-Gibbons~\cite{DG2023} defined the $s$-hull of a lattice, and they show the reduction of LIP of lattices obtained from the code over finite fields to $\mathbb{Z}$LIP and the graph isomorphism problem \rev{for graphs with vertex size $2n$} using the $s$-hull\rev{, where $n$ is the lattice dimension.} 
It is known that $\mathbb{Z}$LIP can be solved in $2^{\frac{0.292n}{2}+o(n)}$ time by the Blockwise Korkine-Zolotarev (BKZ) algorithm~\cite{DW2022}.
Additionally, the graph isomorphism problem can be solved in quasipolynomial time~\cite{B2016}, and the time complexity of this LIP is as long as that of $\mathbb{Z}$LIP. 
This result also serves as a counterexample to the LIP hardness conjecture~\cite{DW2022} using the gap, which is one of the lattice invariants.
Currently, as a conjecture on the difficulty of LIP, one involving a new lattice invariant called the hullgap is being considered~\cite{DG2023}.
% 現在では、LIPの困難性の予想として、hullgapという格子の新たな不変量を用いたものが考えられている~\cite{DG2023}。

\subsection{Our results}

The results shown by Ducas-Gibbons are only for the lattices obtained by the LCD codes over prime fields with positive characteristic.
We generalize the hull attack of Ducas-Gibbons for the codes over rings $\ZZ/k\ZZ$.  
\begin{itemize}
  \itemc When $k$ is odd, the LIP for lattices obtained by the free and LCD codes can be reduced to the $\ZZ$LIP and graph isomorphism problem \rev{for graphs with vertex size $2n$} (\Cref{thm:maina}).
  \itemc When $k$ is an odd prime power, we can assume the code property only with LCD (\Cref{thm:main}).  
  \itemc When $k$ is even and not divisible by $4$, the free and LCD codes obtained by LIP can also be reduced to the $\ZZ$LIP and graph isomorphism problem \rev{for graphs with vertex size $3n$} (\Cref{thm:mainb}).
\end{itemize}
% \rev{
%   The result of Ducas-Gibbons provides a criterion for the hardness of LIP on certain $q$-ary lattices, but only when $q$ is a prime.
%   Our result extends this criterion to the case where $q$ is not divisible by $4$.
%   In particular, the hardness conjecture given by Ducas and van Woerden [8] should be modified even when $q$ is not necessarily prime, 
%   i.e., the gap in the conjecture should be changed to the hullgap given in Section 6 of [7] as well. \com{A}{1}
% }
% \rev{
% Ducas-Gibbons provides a counterexample for conjecture of the hardness of LIP given by Ducas and van Woerden \cite{DW2022}.
% This counterexample is certain $q$-ary lattices, when $q$ is a prime.  
% Therefore, they modified the conjecture given in Section 6 of \cite{DG2023}, and these lattices do not constitute counterexamples under the modified conjecture.
% Our result extends the construction of such a $q$-ary lattices when $q$ is not divisible by $4$.
% All of these lattices only act the counterexamples under the conjecture in \cite{DW2022} but do not constitute counterexamples under the modified conjecture in \cite{DG2023}.
% This implies the validity of the new conjecture and thus contributes to the security analysis of the LIP.
% }
Ducas and Gibbons provided a counterexample to the conjecture regarding the hardness of LIP proposed by Ducas and van Woerden~\cite{DW2022}.  
This counterexample involves certain $q$-ary lattices when $q$ is a prime.
Accordingly, they define a new lattice invariant, called the \emph{hullgap}.  
Using the hullgap, they modify the conjecture in Section~6 of~\cite{DG2023}, under which these lattices are no longer counterexamples.
Our results extend the construction of such $q$-ary lattices to the case where $q$ is not divisible by $4$.  
All of these lattices serve as counterexamples only to the original conjecture in~\cite{DW2022}, but not to the modified version in~\cite{DG2023}.  
This supports the validity of the revised conjecture and contributes to the security analysis of LIP.

\rev{\subsection{Technical novelty}
To generalize to codes over rings $\ZZ/k\ZZ$, there are two obstacles.
% One is the existence of a projection from $\ZZ^k$ onto the code.
% If $k$ is a prime power, then such a projection exists if and only if the code is LCD.
% However, in other cases of $k$, we need to assume an additional restriction on the code: the free property.
One is that we can not use the result of \cite{BAA2019} directly.
To overcome this, we generalize the result of \cite{BAA2019} for free and LCD codes over $\ZZ/k\ZZ$.
The other obstacle is that the element $2 \in \ZZ/k\ZZ$ is often not a unit.
To overcome this, we define the extended signed closure of the code and obtain the result when $k$ is not divisible by $4$.
At the beginning of \Cref{sec:4}, we provide a more detailed description of the technical overview.
\com{A}{2}
% Similarly to the previous results, the LCD property of codes plays an important role in this result.
% Additionally, when generalizing the previous results to finite rings, the free property of the codes becomes important.
}
% Ducas-Gibbonsによるhull攻撃の結果は、位数が素数であるような有限体上のLCD符号から得られる格子に対してのみ適用できるものだった。
% 本稿では、Ducas-Gibbonsによるhull攻撃の結果を、有限環$\ZZ/k\ZZ$から得られる符号に対して一般化する。
% % $k$が素数の冪である場合、Ducas-Gibbonsの結果と全く同様の結果を得ることが出来た。
% $k$が奇数である場合は、自由かつLCD符号であればDucas-Gibbonsの結果と同様に、格子の同型問題が$\ZZ$LIPと$2n$頂点の重み付きグラフ同型問題へと帰着できることがわかった(\Cref{thm:maina})。
% 特に、$k$が奇素数の冪である時はLCD符号という仮定のみで良いことが判明した(\Cref{thm:main})。
% 更に、$k$が$4$の倍数でない偶数の場合も、自由かつLCD符号から得られる格子の格子同型問題が$\ZZ$LIPとグラフ同型問題へと帰着できた(\Cref{thm:mainb})。
% ただし、この場合は帰着されるグラフの頂点数が$3n$頂点となった。

\subsection{Related work}
As in LWE-based lattice cryptography, there exist two types of attacks for LIP-based cryptosystems: attacks for {\em general} LIP problems and ones for {\em structured} LIP problems. 
For the latter case, we see a series of tremendous progresses recently~\cite{mpw24,cfmpw25,apw25}.
However, the attacks have not yet been succeeded to break Hawk, which is the main target for the structured LIP attack. 
In contrast to them, our attack target is general LIP as is given in~\cite{DG2023}.

\subsection{Outline}
This paper is organized as follows.  
In \Cref{sec:2}, we provide definitions of lattice, code, and graph, and introduce \rev{computational} problems related to them.  
In \Cref{sec:3}, we present the previous work on the hull attack by Ducas-Gibbons.  
In \Cref{sec:4}, we generalize the hull attack for lattices constructed by codes in the finite ring $\ZZ/k\ZZ$ and consider the case when $k$ is odd and when $k$ is even but not divisible by $4$.

% 本稿の構成は、$2$節で格子、符号、グラフの数学的な定義と関連する計算問題を述べる。
% $3$節で先行研究であるDucas-Gibbonsによるhull攻撃の結果を述べる。
% $4$節ではhull攻撃を有限環$\ZZ/k\ZZ$上の符号から得られる格子に対して一般化し、$k$が奇数の場合と$k$が$4$の倍数でない偶数の場合を考察する。

\section{Preliminary}\label{sec:2}
\subsection{Lattices}

Let $L$ be a discrete subgroup of an $n$-dimensional vector space over $\RR$.
If there exist $\vv_1, \ldots, \vv_m$ such that 
\[
  L \coloneq \left\{ \sum_{i=1}^{m} x_i \vv_i \mid x_i \in \ZZ \right\},
\]
then $L$ is called a \textbf{lattice}.
The \textbf{dual of the lattice} $L$, denoted by $L^*$, is defined as follows:
\[
  L^* \coloneq \{ \xx \in \RR^n \mid \xx \cdot \yy \in \ZZ\text{ for all $\yy \in L$}\}.
\]
\rev{Let $s$ be a real number.
Then, we call $L \cap sL^*$ the \textbf{$s$-hull} of $L$ and denote it by $H_s(L)$.\com{A}{10}
}
% 実数$\RR$の$n$次元ベクトル空間の部分集合$L$が、$n$個の一次独立なベクトル$\vv_1,\ldots,\vv_m$を使って
% \[
%   L\coloneq\left\{\sum_{i=1}^{m}x_i\vv_i\mid x_i\in\ZZ\right\},
% \]
% と書ける時、$L$を\textbf{格子}と呼ぶ。
% 更に、格子$L$に対してその\textbf{双対格子}$L^*$を
% \[
%   L^*\coloneq\{\xx\in\RR^n\mid \xx\cdot\yy\in\ZZ,\yy\in L\}
% \]
% とする。

\begin{df}[Lattice Isomorphism]
  Lattices $L_1$ and $L_2$ are isomorphic if there exists an orthonormal matrix $O$ such that $L_1 = O L_2$.
\end{df}

\rev{\begin{df}[Lattice Isomorphism Problem]
  \textbf{The Decision-Lattice Isomorphism Problem (Decision-LIP)} is to determine whether $L_1$ and $L_2$ are isomorphic, 
  and \textbf{the Search-Lattice Isomorphism Problem (Search-LIP)} is to find the orthonormal matrix $O$ that transforms $L_1$ into $L_2$.
  In particular, when $L_2 = \mathbb{Z}^n$, this problem is called $\mathbb{Z}$LIP.
  In this paper, we focus only on the Search-LIP and simply denote LIP as this second problem.\com{A}{4,5}
\end{df}}

% \rev{\begin{df}[Lattice Isomorphism Problem]
%   \textbf{Lattice isomorphism problem (LIP)} is to distinguish whether $L_1$ and $L_2$ are isomorphic, and to determine the orthonormal matrix $O$ that transforms $L_1$ to $L_2$.
%   In particular, when $L_2 = \mathbb{Z}^n$, this problem is called $\mathbb{Z}$LIP.  
% \end{df}}

\rev{\begin{rem}
%   There are two types of problems regarding lattice isomorphism.  
% One is the problem defined above, which is also called the Search-Lattice Isomorphism Problem (Search-LIP).  
% The other is the Decision-Lattice Isomorphism Problem, which asks whether two given lattices are isomorphic.  
In general, the Code Equivalence Problem (CEP) and the Graph Isomorphism Problem (GI), which are defined later, also have the two types: the search version and the decision version.  
In this paper, we \rev{mainly focus on} the search versions of these problems.  
Thus, hereinafter, when we refer to LIP, CEP, or GI, we mean the corresponding search problems.
\end{rem}\com{A}{5}}
\subsection{Codes}

Let $R$ be a finite ring and let $R^n$ be a free module over $R$.
Then, the linear code $C$ over $R$ with length $n$ is a submodule of $R^n$.
Especially, if $R$ is a finite field of order $q$ and the dimension of $C$ is $m$, then $C$ is called an \textbf{$[n,m]_q$ linear code}.
When a code $C$ over $R$ is denoted as $C = \{cG \mid c \in R^m\}$, where $G$ is a matrix over $R$ with size $m \times n$ and no row of $G$ can be represented as a linear combination of the other rows, then $G$ is called the \textbf{generator matrix} of $C$.
A dual code of $C$, denoted as $C^{\perp}$, is defined as follows:
\[
  \rev{C^\perp\coloneq\{\xx\in R^n \mid \xx\cdot \cc=0\text{ for all $\cc\in C$ }\}.}
\]
$C \cap C^\perp$ is called the \textbf{hull}, and if the hull is trivial, meaning the hull is $\{0\}$, then $C$ is called an \textbf{LCD code} (linear complementary dual code).

% 有限環$R$に対して、$R$の自由加群を$R^n$と書いたとき、$R$上の長さ$n$の線形符号とは、$R^n$の部分加群のことである。
% 特に、位数$q$の有限体$\FF_q$の$n$次元ベクトル空間の$m$次元部分空間$C$を\textbf{$[n,m]_q$線形符号}と呼ぶ。
% 符号$C$の\textbf{双対符号}$C^\perp$とは、
% \[
%   C^\perp\coloneq\{\xx\in \FF_q^n\mid \xx\cdot \cc=0, \cc\in C\}
% \]
% で定義される線形空間のことである。
% また、符号$C$に対して$C\cap C^\perp$を$C$の\textbf{hull}と呼び、hullが自明、すなわち$\{0\}$の時、符号$C$を\textbf{LCD符号}(linear complementary dual)と呼ぶ。

\begin{df}[Code Equivalence]
  Codes $C_1$ and $C_2$ are called equivalent if there exist a diagonal matrix $D$ and a permutation matrix $P$ such that the following holds:
  \[
  C_1 = DPC_2.
  \]
  In particular, when all the elements of $D$ are $\pm 1$, this equivalence is called \textbf{signed permutation code equivalence}, and if $D$ is the identity matrix, it is called \textbf{permutation code equivalence}.
\end{df}

\rev{
\begin{df}[Code Equivalence Problem]
  The \textbf{code equivalence problem} is the problem of determining the equivalence of $C_1$ and $C_2$, and, when they are equivalent, finding the matrices $D$ and $P$.
  In particular, when $C_1$ and $C_2$ are signed permutation code equivalent, this problem is called the \textbf{signed permutation code equivalence problem (SPEP)}, and when $C_1$ and $C_2$ are permutation code equivalent, it is called the \textbf{permutation code equivalence problem (PEP)}.
  \com{A}{4}
\end{df}
}% \begin{df}[符号同型、符号同型問題]
  % 符号$C_1,C_2$が同型であるとは、下記の条件を満たす対角行列$D$と置換行列$P$が存在することである。
  % \[
  %   C_1=DPC_2
  % \]
  % 特に、$D$の要素が$\pm1$のみである時、これを\textbf{サイン付き置換同型}と呼び、$D$が単位行列である時\textbf{置換同型}と呼ぶ。
%   \textbf{符号同型問題}とは、二つの符号$C_1,C_2$が同型であるかを判定し、同型である場合に対角行列$D$と置換行列$P$を決定する問題である。
%   特に、$D$の対角成分が$\pm1$のみであることがわかっている場合を符号の\textbf{サイン付き置換同型問題(SPEP)}と呼び、$D$が単位行列であることがわかっている場合を符号の\textbf{置換同型問題(PEP)}と呼ぶ。
% \end{df}

In general, the code equivalence problem of codes over a finite field $\FF_q$ with length $n$ is equivalent to the permutation code equivalence of closure codes, which are obtained from the original codes and are over a finite field $\FF_q$ with length $(q-1)n$~\cite{ss13a, ss13b}.
Hereinafter, \rev{let $\ZZ_k$ denote $\ZZ/k\ZZ$}, and we consider the codes over a finite ring $\ZZ_k$.
\rev{One of the methods} of obtaining a lattice from codes over $\ZZ_k$ is called \textbf{Construction A}.
Let $\iota_k$ be a function from $\ZZ$ to $\ZZ_k$ such that if an integer $x$ is congruent to $y$ modulo $k$, then $\iota_k(x)=y$.
Additionally, we extend the domain of $\iota_k$ to $\ZZ^n$ such that $\iota_k(\vv) = (\iota_k(v_1), \ldots, \iota_k(v_n))$, where $\vv = (v_1, \dots, v_n)$.
Then, the lattice $C + k\ZZ^n$ obtained from the code $C$ using Construction A is defined as follows~\cite{cs99}:
\[
  C + k\ZZ^n \coloneqq \{\xx \in \ZZ^n \mid \iota_k(\xx) \in C\}.
\]

% 一般に$\FF_q$上の長さ$n$の符号の同型問題は、与えられた符号から得られる、閉包符号と呼ばれる$\FF_q$上の長さが$(q-1)n$の符号の置換同型問題と同値であることが知られている~\cite{ss13a,ss13b}。
% 以降、$\ZZ_k=\ZZ/k\ZZ$として、$\ZZ_{k}$上の符号を考える。
% 環$\ZZ_k$上の符号から格子を与える方法として、\textbf{構成法A}と呼ばれる手法がある。
% $\ZZ$から$\ZZ_k$への関数$\iota_k$を、ある整数$x$を$p$で割った余りが$y$の時に、$\iota_k(x)=y$となるような写像とする。
% 更に、$\ZZ$上のベクトル$\vv=(v_1,\ldots,v_n)$に対して、$\iota_k(\vv)=(\iota_p(v_1),\ldots,\iota_p(v_n))$と定義する。
% この時、$C$の構成法Aによって得られる格子$C+k\ZZ^n$とは、
% \[
%   C+k\ZZ^n\eqcolon\{\xx\in\ZZ^n\mid \iota_k(\xx)\in C\}
% \]
% で定義される格子である。

\subsection{Graphs}

\rev{A graph $G=(V,E)$ consists of a finite set $V$ of vertices and a set $E$ of edges, 
where each edge is a 2-element subset of $V$.}
If there exists a mapping $w$ from $E$ to $\mathbb{R}$, then $(V,E,w)$ is called a weighted graph. 
% \rev{In this paper, we focus only on undirected graphs, that is, if $(u,v)\in E$, then $(v,u)$ is also in $E$.
% Therefore, we denote an edge as a subset of $V$ with $2$ elements.\com{A}{6}}
The adjacency matrix of $(V,E,w)$ is a $|V| \times |V|$ matrix with entries
\[
A_{ij} =
\begin{cases}
  w(\{i,j\}) & \text{if } \{i,j\} \in E, \\
  0 & \text{if } \{i,j\} \notin E.
\end{cases}
\]

% ある集合$V$とその二点集合の集合$E$があった時、その組$(V,E)$をグラフと呼ぶ。
% また、$E$から$\RR$への写像$w$が与えられている時、$(V,E,w)$を重み付きグラフと呼ぶ。
% 重み付きグラフ$(V,E,w)$の隣接行列$A$とは、グラフの頂点集合$V$によって添え字づけられた$|V|\times|V|$行列で、
% \[A_{ij}=\begin{cases}
%   w(\{i,j\})&\{i,j\}\in E\\
%   0 & \{i,j\}\notin E
% \end{cases}\]
% によって定義される行列である。

\begin{df}[Graph Isomorphism]
  Let $G_1 = (V_1, E_1, w_1)$ and $G_2 = (V_2, E_2, w_2)$ be weighted graphs, and let $A_1$ and $A_2$ be the adjacency matrices of $G_1$ and $G_2$, respectively. 
  $G_1$ and $G_2$ are called isomorphic if there exists a permutation matrix $P$ such that $A_1 = P^\top A_2 P$ holds. 
\end{df}

\rev{\begin{df}[Graph Isomorphism Problem]
  \textbf{Graph isomorphism problem (GI)} involves distinguishing whether two given graphs are isomorphic, and if they are, determining the permutation matrix $P$.\com{A}{4}
\end{df}
}
% \begin{df}[グラフ同型問題]
%   二つの重み付きグラフ$G_1=(V_1,E_1,w_1)$と$G_2=(V_2,E_2,w_2)$が同型であるとは、
%   $G_1,G_2$の隣接行列$A_1,A_2$に対して、置換行列$X$で$A_1=X^\top A_2X$となるようなものが存在することである。
%   \textbf{重み付きグラフ同型問題(GI)}とは、二つのグラフ$G_1,G_2$が同型であるかを判定し、同型ならば置換行列$X$を決定する問題である。
%   本稿では常に重み付きグラフを考えるから、単にグラフ同型問題と呼ぶ。
% \end{df}
% % $V_1$から$V_2$への全単射$\phi$で、$\{i,j\}\in E_1$の時かつその時に限り$\{\phi(i),\phi(j)\}\in E_2$であり、$w_1(\{i,j\})=w_2(\{\phi(i),\phi(j)\})$となるような$\phi$が存在することである。
GI is known to be solvable in quasipolynomial time~\cite{B2016}.
Additionally, the permutation code equivalence problem of LCD codes over finite fields can be reduced to GI~\cite{BAA2019}.

% グラフの同型問題は準多項式時間で計算できることが知られている~\cite{B2016}。
% また、有限体上のLCD符号の置換同型問題はグラフ同型問題に帰着できることが知られている~\cite{BAA2019}。

\section{The hull attack proposed by Ducas-Gibbons}\label{sec:3}

% First, we define the hull of a lattice.
\rev{Ducas-Gibbons~\cite{DG2023} investigates the LIP of lattices obtained by Construction A of an LCD code over a finite field. 
They show that the LIP can be reduced to the $\ZZ$LIP and GI using the $p$-hull of the lattice, $H_p(L)$.}
In this section, we present the flow of this reduction and \rev{introduce their lemmas for comparing with our results.}
Hereinafter, let $p$ be a prime number.

\subsection{Outline of the hull attack~\cite{DG2023}}

\begin{figure}[t]
  \begin{tikzpicture}[node distance=2.5cm, auto]

  % ノードの定義
  \node (X) at (0,0) {LIP};
  \node (Y) at (2.5,0.25) {$\mathbb{Z}$LIP};
  \node (Z) at (2.5,-0.25) {SPEP};
  \node (W) at (5,0.25) {$\mathbb{Z}$LIP};
  \node (R) at (5,-0.25) {PEP};
  \node (S) at (7.5,0.25) {$\mathbb{Z}$LIP};
  \node (Wr) at (7.5,-0.25) {GI};

  % 矢印の描画
  % \draw[->] (X) -- (Y) node[midway, above left] {\Cref{lem:9}};
  \draw[->] (X) -- (2.25,0) node[midway, above] {{\tiny \Cref{lem:9}}};
  \draw[->] (2.75,0) -- (4.75,0) node[midway, below] {{\tiny \Cref{lem:10}}} node[midway, below, yshift=-6.5pt]{{\tiny \Cref{lem:11}}};
  % \draw[->] (2.5,0) -- (4.5,0)  {\Cref{lem:11}};
  \draw[->] (5.25,0) -- (7.25,0) node[midway, below] {{\tiny \Cref{thm:1}}};
  % \draw[->] (R) -- (Wr) node[midway, below] {\Cref{thm:1}};

  \end{tikzpicture}
%   \begin{tikzpicture}[node distance=2cm, >=latex]
%   % ノード定義（箱サイズを調整）
%   \node[draw, minimum width=1.2cm, minimum height=1cm, align=center] (A) {ABCD\\EFGH};
%   \node[draw, minimum width=1.2cm, minimum height=1cm, align=center, right of=A] (B) {IJKL\\MNOP};
%   \node[draw, minimum width=1.2cm, minimum height=1cm, align=center, right of=B] (C) {QRST\\UVWX};
%   \node[draw, minimum width=1.2cm, minimum height=1cm, align=center, right of=C] (D) {YZAB\\CDEF};

%   % 矢印とラベル
%   \draw[->] (A) -- node[above] {Top 1} node[below] {Bottom 1} (B);
%   \draw[->] (B) -- node[above] {Top 2} node[below] {Bottom 2} (C);
%   \draw[->] (C) -- node[above] {Top 3} node[below] {Bottom 3} (D);
% \end{tikzpicture}
  \caption{Reductions of computational problems in the hull attack~\cite{DG2023}} \label{fig:1}
\end{figure}
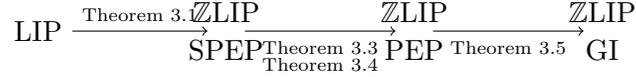

\Cref{fig:1} describes the reductions of computational problems in the hull-attack by Ducas-Gibbons~\cite{DG2023}. 

\rev{\Cref{lem:9} implies that the LIP of specific lattices can be reduced to the $\ZZ$LIP and SPEP of LCD codes with length $n$.
\Cref{lem:10} and \Cref{lem:11} suggest that this SPEP can also be reduced to the PEP of LCD codes with length $2n$.}
Additionally, using \Cref{thm:1}, we can reduce this PEP to GI.
As a result, for any $p$, the LIP of lattices obtained by Construction A from LCD codes over $\FF_p$ can be reduced to the $\ZZ$LIP and GI (\Cref{fig:1}).
\rev{Note that \Cref{lem:11} treats with the case that $p\neq2$. 
But when $p=2$, SPEP is equivalent to PEP, so by the same argument, this reduction also holds.}

\subsection{\rev{Lemmas for the hull attack\com{B}{4}}}

We introduce some lemmas and terms that are used in the hull attack.

\begin{lem}[Lemma $9$~\cite{DG2023}]\label{lem:9}
  Let $C$ be an \rev{$[n,k]_p$ LCD code.}
  For any $i \in \{1,2\}$, let $O_i$ be an orthonormal matrix of size $n$, and let $L_i = O_i(C + p\mathbb{Z}^n)$.
  Then, the LIP of $L_1$ and $L_2$ can be reduced to the $\mathbb{Z}$-LIP and SPEP of LCD codes with length $n$.
\end{lem}
% \begin{lem}[補題$9$~\cite{DG2023}]\label{lem:9}
%   $C$を$[n,k]_p$符号で、LCD符号とする。
%   また、任意の$i\in\{1,2\}$に対して$O_i$を$n$次の実数上の直交変換として、$L_i=O_i(C+p\ZZ^n)$とする。
%   この時、$L_1$と$L_2$の格子同型問題は、$\ZZ$LIPとLCD符号の符号付き置換同型問題へと帰着できる。
% \end{lem}

\begin{df}[Signed closure code~\cite{DG2023}]
  Let $C$ be an $[n,k]_p$ code.
  The signed closure code of $C$, denoted as $C^\pm$, is a $[2n,k]_p$ code and is defined as follows:
  \[
    C^\pm \coloneq \{(x_1, -x_1, x_2, -x_2, \ldots, x_n, -x_n) \mid (x_i)_{1 \leq i \leq n} \in C\}
  \]
\end{df}

% \begin{df}[サイン閉包符号~\cite{DG2023}]
%   $[n,k]_p$符号$C$について、そのサイン閉包符号$C^{\pm}$とは、$[2n,k]_p$符号であり、
%   \[C^\pm\coloneq\{(x_1,-x_1,x_2,-x_2,\ldots,x_n,-x_n)\mid (x_i)_{1\leq i\leq n}\in C\}\]
%   で定義される符号である。
% \end{df}

\begin{lem}[Lemma $10$~\cite{DG2023}]\label{lem:10}
  Let $C$ and $C'$ be $[n,k]_p$ codes. 
  Then, $C$ and $C'$ are signed permutation code equivalent if and only if $C^\pm$ and $C'^\pm$ are permutation code equivalent.
\end{lem}

\begin{lem}[Lemma $11$~\cite{DG2023}]\label{lem:11}
  If $p \neq 2$, $C$ is LCD if and only if $C^\pm$ is LCD.
\end{lem}

\begin{lem}[Theorem $5$~\cite{BAA2019}]\label{thm:1}
  Let $C_1$ and $C_2$ be LCD codes over a finite field $\FF_p$, and let $G_1$ and $G_2$ be the generator matrices of $C_1$ and $C_2$, respectively.
  Then, $C_1$ and $C_2$ are permutation code equivalent if and only if the two symmetric matrices $G_1^\top (G_1 G_1^\top)^{-1} G_1$ and $G_2^\top (G_2 G_2^\top)^{-1} G_2$ are graph isomorphic.
\end{lem}

% Next, we briefly \rev{sketch\com{B}{11}} the proof of the lemmas.
\Cref{lem:9} is obtained from the property of the $p$-hull of a lattice, as shown in the following lemma.

\begin{lem}[Lemma $4$~\cite{DG2023}]\label{lem:4}
Let $C$ be an $[n,k]_p$ code and let $L = C + p\ZZ^n$ be a lattice. Then,
\[
  H_p(L) = C \cap C^\perp + p\ZZ^n.
\]
\end{lem}

% \begin{lem}[補題$4$~\cite{DG2023}]\label{lem:4}
% $C$を$[n,k]_p$符号として、格子$L$を$L=C+p\ZZ^n$とする。
% この時、
% \[
%   H_p(L)=C\cap C^\perp +p\ZZ^n.
% \]
% \end{lem}

The proof of \Cref{lem:10} follows the same argument as in~\cite{ss13a}, and \Cref{lem:11} can be proven by considering the inner product of two vectors in the code.

% \rev{Finally, we present the flow of the reduction in \Cref{thm:1}.\com{B}{11}}
In~\cite{BAA2019}, it was shown that given an LCD code $C$ and its generator matrix $G$, the matrix $G^\top(GG^\top)^{-1}G$ is a \rev{projection} from $\FF_q^n$ to $C$.
\rev{Using this matrix, the PEP of the LCD code is reduced to GI and obtain \Cref{thm:1}.\com{B}{4}}
Note that if $C$ is over a finite field, then $C$ is LCD if and only if $GG^\top$ is non-singular.
Hence, we can always define $G^\top(GG^\top)^{-1}G$.
\rev{However, in general, if $C$ is defined over $\ZZ_k$, this does not hold. For the case of our finite ring, see \Cref{thm:1a} and its proof.}
% \Cref{lem:10}の証明は、~\cite{ss13a}の補題$4$と同様の手法で与えられており、\Cref{lem:11}は符号内の二つのベクトルの内積を考えることですぐにわかる。

% 最後に、\Cref{thm:1}の証明を簡単に紹介する。
% ~\cite{BAA2019}では、LCD符号$C$の生成行列$G$を考えた時、$G^\top(GG^\top)^{-1}G$という対称行列が$\FF_q^n$から$C$への射影となることを用いて、LCD符号の置換同型問題をグラフ同型問題へと帰着させた。
% 注意として、有限体上の符号であれば、$C$がLCDであることと$GG^\top$が正則行列となることは同値であることが知られているから、LCD符号であれば上記の射影は常に定義できる。

\section{Our results}\label{sec:4}

The previous work of Ducas and Gibbons considers lattices obtained by Construction A from codes over finite fields.  
In this section, we generalize this result to lattices obtained from codes over the finite ring $\ZZ_k$.
To generalize for the codes over $\ZZ/k\ZZ$, one of the obstacles is the existence of a projection from $\ZZ^k$ onto the code.
\Cref{thm:1a} shows that if a code is free and LCD, then $G^\top(GG^\top)^{-1}G$ is such a projection, where $G$ is a generator matrix of a code.
% If this projection exists, then we can obtain the reduction in a manner similar to that of Ducas-Gibbons(\Cref{thm:maina}).
Another obstacle is that $2$ is not a unit when $k$ is even.
To resolve this, we define the extended signed closure of a code (\Cref{df:esc}) and prove the reduction of SPEP of codes to PEP of the extended signed closure of codes when $k$ is not divisible by $4$ (\Cref{lem:10b}).
Moreover, to prove \Cref{lem:10b}, we also provide a proof of \Cref{lem:10a} that closely follows the structure of the proof of \Cref{lem:10}, but introduces permutation matrices, which were not used in the original proof.
  In \Cref{ssec:4-1}, we prove that when $k$ is odd, LIP for lattices constructed from codes over $\ZZ_k$ can be reduced in a manner similar to that of Ducas-Gibbons (\Cref{thm:maina}), and we also focus on the special case when $k$ is a power of an odd prime (\Cref{thm:main}).
  In \Cref{ssec:4-2}, we prove a similar reduction when $k$ is even and not divisible by $4$ (\Cref{thm:mainb}).

\subsection{When $k$ is odd}\label{ssec:4-1}

To generalize the results of Ducas-Gibbons to $\ZZ_k$ codes, we need \rev{some} restrictions on the codes.
\begin{df}[Free code]
  Let $R$ be a ring and let $C$ be a linear code over $R$ with generator matrix $G$.
  $C$ is called a \textbf{free code} if all the row vectors of $G$ are linearly independent over $R$.
\end{df}
Note that $C$ is free and LCD if and only if $GG^\top$ is non-singular.
Additionally, $C$ is always free if $C$ is a code over a finite field.
When $k$ is odd, we have the following theorem.
\begin{thm}\label{thm:maina}
  Let $k$ be an odd integer and let $C$ be a free and LCD code over $\ZZ_{k}$ with length $n$.
  For any $i \in \{1,2\}$, let $O_i$ be an orthonormal matrix of size $n$, and let $L_i = O_i(C + k\ZZ^n)$.
  Then, the LIP of $L_1$ and $L_2$ can be reduced to $\ZZ$LIP and GI with vertex size $2n$.
\end{thm}
To prove \Cref{thm:maina}, we generalize \Cref{lem:9} - \ref{thm:1} to $\ZZ_k$ codes, which are \Cref{lem:9a} - \ref{thm:1a}.
In fact, \Cref{lem:9} and \Cref{lem:10} can be generalized to $\ZZ_k$ codes for all $k$.
\begin{lem}\label{lem:9a}
  For any integer $k$, let $C$ be an LCD code over $\ZZ_{k}$ with length $n$.
  For any $i \in \{1, 2\}$, let $O_i$ be an orthonormal matrix of size $n$, and let $L_i = O_i(C + k\ZZ^n)$.
  Then, the LIP of $L_1$ and $L_2$ can be reduced to the $\ZZ$-LIP and the SPEP for free and LCD codes.
\end{lem}
\begin{lem}\label{lem:10a}
  Two linear codes $C_1$ and $C_2$ over $\ZZ_k$ with length $n$ are signed permutation code equivalent if and only if $C_1^\pm$ and $C_2^\pm$ are permutation code equivalent.
\end{lem}
On the other hand, we have to restrict $k$ to be odd to generalize \Cref{lem:11}.
\begin{lem}\label{lem:11a}
  When $k$ is odd, a linear code $C$ over $\ZZ_k$ with length $n$ is free and LCD if and only if $C^\pm$ is free and LCD.
\end{lem}
Furthermore, to generalize \Cref{thm:1} to $\ZZ_k$, the code must be free.
\begin{lem}\label{thm:1a}
Let $C_1$ and $C_2$ be free and LCD codes over $\ZZ_k$ with length $n$ and generator matrices $G_1$ and $G_2$, respectively.
Then, $C_1$ and $C_2$ are permutation code equivalent if and only if $G_1^\top(G_1G_1^\top)^{-1}G_1$ and $G_2^\top(G_2G_2^\top)^{-1}G_2$ are graph isomorphic.
\end{lem}
Note that $k$ can be any integer in \Cref{thm:1a}.

Using these lemmas, we can prove \Cref{thm:maina}.

\begin{pproof}{\Cref{thm:maina}}
\rev{From \Cref{lem:9a}, the LIP for lattices obtained by Construction A from any LCD code over $\ZZ_k$ with length $n$ can be reduced to the $\ZZ$LIP and SPEP for LCD codes over $\ZZ_k$ with length $n$. 
  From \Cref{lem:10a}, this SPEP can be reduced to the PEP for codes over $\ZZ_k$ with length $2n$. 
  Additionally, from \Cref{lem:11a}, if $k$ is odd, the code obtained by \Cref{lem:10a} preserves the properties of being free and LCD. 
  Finally, from \Cref{thm:1a}, the PEP for free and LCD codes can be reduced to the graph isomorphism problem, yielding \Cref{thm:maina}.\com{A}{9}
}\end{pproof}

\subsubsection{Proofs of \Cref{lem:9a} - \ref{thm:1a}}
Next, we prove \Cref{lem:9a} - \ref{thm:1a} using a similar discussion to that of Ducas-Gibbons.
Firstly, we generalize \Cref{lem:4} to a linear code over $\ZZ_k$.

\begin{lem}\label{lem:4a}
  Let $C$ be a linear code over $\ZZ_{k}$ with length $n$, and let $L=C+k\ZZ^n$ be the lattice obtained from the Construction A of $C$.
  Then,
  \[
    H_{k}(L)=C\cap C^\perp +k\ZZ^n.
  \]
\end{lem}

\begin{proof}
  First, we show 
  \begin{equation}\label{eq:1}
    C^\perp+k\ZZ^n=k(C+k\ZZ^n)^*.
  \end{equation}
  For all elements $x \in C^\perp + k\ZZ^n$ and $y \in C + k\ZZ^n$, the inner product $x \cdot y$ is in $k\ZZ$.
  Therefore, $C^\perp + k\ZZ^n \subset k(C + k\ZZ^n)^*$.
  Similarly, for all $z \in k(C + k\ZZ^n)^*$ and $y \in C + k\ZZ^n$, from the definition of the dual of a lattice, 
  $z \cdot y \in k\ZZ$.
  In particular, $\iota_k(z) \cdot \iota_k(y) \in k\ZZ$ and $\iota_k(z) \in C^\perp$ imply $\iota_k(k(C + k\ZZ^n)^*) \in C^\perp$.
  Therefore, $C^\perp + k\ZZ^n \supset k(C + k\ZZ^n)^*$ and we obtain \Cref{eq:1}.
  Hence,
  \begin{align*}
    H_{k}(L) &= (C + k\ZZ^n) \cap k(C + k\ZZ^n)^* \\
    &= (C + k\ZZ^n) \cap (C^\perp + k\ZZ^n) \\
    &= (C \cap C^\perp) + k\ZZ^n.
  \end{align*}
\end{proof}

We can prove \Cref{lem:9a} using \Cref{lem:4a}.

\begin{pproof}{\Cref{lem:9a}}
  Firstly, we calculate $k$ from the given lattices, $L_1$ and $L_2$.
  Let $G_i$ be the generator matrix of $L_i$ with size $m \times n$, then the determinant of $G_i G_i^\top$ is $k^{n-m}$.
  Therefore, we obtain $k$ by computing the determinant of $G G^\top$.
  Next, we consider the $k$-hull of $L_i$, denoted as $H_k(L_i)$.
  From \Cref{lem:4a} and the assumption that $C$ is LCD, we have
  \begin{align*}
    H_k(L_i) &= O_i(H_k(C + k \ZZ^n)) \\
             &= O_i((C \cap C^\perp) + k \ZZ^n) \\
             &= k O_i \ZZ^n.
  \end{align*}
  Hence, if we solve $\ZZ$LIP for lattices $H_k(L_i)$, then we obtain the orthonormal matrix $\hat{O}_i$, which transforms $H_k(L_i)$ to $k \ZZ^n$.
  % Now, $\hat{O_i}$ is a product of $O_i^{-1}$ and the automorphism of $\ZZ^n$, which is represented by the product of the diagonal matrix $D_i$ and the permutation matrix $P_i$, where all entries of $D_i$ are $\pm 1$.
  % This implies that $\hat{O_i} = D_i P_i O_i^{-1}$.
  After this point, the reduction is exactly the same as the proof of \rev{\Cref{lem:9}}.
  Since $\hat{O}_i$ is a product of $O_i^{-1}$ and an automorphism of $\ZZ^n$, there exists a diagonal matrix $D_i$, whose entries are all $\pm 1$, and a permutation matrix $P_i$ such that $\hat{O}_i = D_i P_i O_i^{-1}$.
  Therefore, by applying $\hat{O}_i$ to $L_i$, we obtain $\hat{O}_i L_i = D_i P_i (C + k \ZZ^n) = D_i P_i C + k \ZZ^n$.
  Now, the problem of determining $DP$ such that $D_1 P_1 C = DP D_2 P_2 C$ is an SPEP.
  As a whole, since $\hat{O}_i$ is obtained by $\ZZ$LIP and $DP$ is computed by SPEP, the proof is complete.
\end{pproof}

\Cref{lem:10a} and \Cref{lem:11a} are also obtained by the same arguments which were given in Ducas-Gibbons~\cite{DG2023}.
While the proof of \Cref{lem:10a} is essentially the same as that of Ducas and Gibbons, we provide a proof using permutation matrices in preparation for the proof of \Cref{lem:10b}.

\begin{pproof}{\Cref{lem:10a}}
  Let $i\in\{1,2\}$, let $G_i$ be the generator matrix of $C_i$, and let $T$ be a matrix as follows:
  {\small
  \[
    T = \begin{pmatrix}
      1 & -1 & &&&&& \\
      && 1 & -1 & &&& \\
      &&&& \ddots &&& \\
      &&&&&& 1 & -1
    \end{pmatrix}.
  \]
  }
  Then, the generator matrix of $C_i^\pm$, denoted as $G_i^\pm$, satisfies $G_i^\pm = G_i T$.
  Now, if $C_1$ and $C_2$ are signed permutation code equivalent, then there exist a permutation matrix $P$ and a diagonal matrix whose entries are only $\pm 1$ such that $G_1 = G_2 P D$.
  Since the matrix $D T$ is a matrix whose rows are either the rows of $T$ or the rows of $T$ multiplied by $-1$, there exists a permutation matrix $P'$ such that $D T = T P'$.
  \rev{By} a similar discussion, $P T$ is a matrix whose rows are equal to some rows of $T$, which means there exists a permutation matrix $P''$ such that $P T = T P''$.
  In fact, the swapping of the $i$-th row and $j$-th row can be represented as the swapping of the $(2i-1)$-th column and $(2j-1)$-th column, and swapping the $2i$-th and $2j$-th columns.
  Hence, there exist $P'$ and $P''$ such that
  \[
    G_1 T = G_2 P D T = G_2 T P'' P',
  \]
  which implies that $C_1^\pm$ are permutation code equivalent to $C_2^\pm$.

  On the contrary, we assume that $C_1^\pm$ are permutation code equivalent to $C_2^\pm$.
  From the definition of $C_i^\pm$, the generator matrices are $G_iT$, and \rev{from their permutation code equivalence}, there exists a permutation matrix $P$ such that $G_1T = G_2TP$.
  Now, $TP$ is a matrix obtained by permuting the column vectors of $T$.
  Therefore, let $G_1T = (u_1, \ldots, u_{2n})$ and $G_2T = (v_1, \ldots, v_{2n})$, then there exists a permutation $\sigma$ corresponding to $P$ such that $u_i = v_{\sigma(i)}$.
  Then, from the same discussion of the proof of \rev{\Cref{lem:10}}, we can assume that there exist $\tau_i$ for any $i\in\{1,\ldots,n\}$ such that 
    \[
      \{\tau_i(2i-1),\tau_i(2i)\}=\{2k-1,2k\}
    \]
    for some $k\in\{1,\ldots,n\}$ and $\sigma=\prod_{i=1}^n\tau_i.$\com{B}{5}
  
%   Next, we show that there exists a permutation $\sigma'$ such that the matrix $P'$ corresponding to $\sigma'$ satisfies $G_1T = G_2TP'$ and for any odd number $i$ less than $2n$,
% % ここら辺から同じ議論だから省く？
%   {\small
%   \[
%   \begin{cases}
%     \sigma'(i+1) = \sigma'(i) + 1 & \text{if } \sigma'(i) \text{ is odd}, \\
%     \sigma'(i-1) = \sigma'(i) - 1 & \text{if } \sigma'(i) \text{ is even}.
%   \end{cases}
%   \]
%   }
%   When $\sigma(i)$ is odd and $\sigma(i+1) \neq \sigma(i) + 1$, from the definition, $u_{i+1} = -u_i$ and $v_{\sigma(i)+1} = -v_{\sigma(i)} = -u_i$, so we obtain $v_{\sigma(i+1)} = v_{\sigma(i) + 1}$.
%   Therefore, if we apply the swap of $\sigma(i+1)$ and $\sigma(i)+1$ to $\sigma$, denoted as $\sigma_1$, then $P'$ corresponding to $\sigma_1$ also satisfies $G_1T = G_2TP'$.
%   Similarly, we can also apply the swap to $\sigma$ when $\sigma(i)$ is even.
%   We can repeat such actions on $\sigma$, and finally, we obtain $\sigma'$ that satisfies the previous conditions.
%   Let $\tau_i$ be a permutation such that for some odd number $2i-1$, $\tau_i(2i-1)=2j-1$ and $\tau_i(2i)=2j$, or $\tau_i(2i-1)=2j,\tau_i(2i)=2j-1$, and it does not swap other entries. 
%   Then, there exists $I \subset \{1, \dots, n\}$ such that $\sigma' = \prod_{i \in I} \tau_i$. 
  Furthermore, if $\tau_i(2i-1) = 2j-1$, then this permutation can be represented as the swap of the $i$-th row and the $j$-th row in $T$, and if $\tau_i(2i-1) = 2j$, then this permutation can also be represented as the swap of the $i$-th row and the $j$-th row in $T$, followed by multiplying $-1$ to the $i$-th and $j$-th rows. 
  This implies that there exists a permutation matrix $P'$ and a diagonal matrix $D$ with all entries being $\pm 1$ such that 
  $TP = P'DT$, where $P$ corresponds to the permutation $\sigma$. Therefore, $G_1T = G_2TP = G_2P'DT$. 
  
  Additionally, let $T'$ be the matrix as follows:
  % {\small
  \[
      T' = \begin{pmatrix}
        1 & & & \\
        0 & & & \\
        & 1 & & \\
        & 0 & & \\
        && \ddots & \\
        && & 1 \\
        && & 0
      \end{pmatrix}.
  \]
  % }
  Then, since $TT'$ is the identity matrix, $G_1 = G_2P'D$, and we obtain that $C_1$ and $C_2$ are signed permutation code equivalent.  \end{pproof}

\begin{pproof}{\Cref{lem:11a}}
  The freeness of $C^\pm$ is obtained from the fact that the generator matrix of this code is $GT$.
  For any vector $\xx = (x_1, x_2, \ldots, x_n)$ and $\yy = (y_1, y_2, \ldots, y_n)$ in $C$, both 
  \begin{align*}
    \xx^\pm &= (x_1, -x_1, \ldots, x_n, -x_n),\\
    \yy^\pm &= (y_1, -y_1, \ldots, y_n, -y_n)
  \end{align*}
  are elements of $C^\pm$, and 
  \[
    \xx^\pm \cdot \yy^\pm = 2 \xx \cdot \yy.
  \]
  Here, since $k$ is an odd \rev{integer}, $2$ is a unit over $\ZZ_k$.
  Hence, $\xx \cdot \yy = 0$ is equivalent to $\xx^\pm \cdot \yy^\pm = 0$.
  From the above, the LCD property of $C^\pm$ is equivalent to that of $C$.
\end{pproof}

\begin{pproof}{\Cref{thm:1a}}
  % ~\cite{BAA2019}と同様な流れで証明する。
  Let $C$ be a free and LCD code over $\ZZ_{k}$ with length $n$, and let $G$ be a generator matrix of $C$ with size $m \times n$. We show that $G^\top (GG^\top)^{-1} G$ is a projection of $\ZZ_{k}^n$ onto $C$. 
  \rev{Since $C$ is LCD , $(C \cap C^\perp)^\perp = \ZZ_{k}^n$.
  Additionally, since $C$ is free, the size of the generator matrix of $C^\perp$, denoted as $H$, is $(n - m) \times n$. \com{B}{3}}
  From Theorem 3.1~\cite{TI2000}, the generator matrix of $\ZZ_{k}^n = (C \cap C^\perp)^\perp = C \oplus C^\perp$ is given by
  \[
    \begin{pmatrix}
      G \\
      H
    \end{pmatrix}.
  \]
  Since $GG^\top$ is a non-singular matrix, we can define the following matrix:
  \[
    \begin{pmatrix}
      G^\top (GG^\top)^{-1} & H^\top (HH^\top)^{-1}
    \end{pmatrix},
  \]
  which is the inverse matrix of $\begin{pmatrix} G \\ H \end{pmatrix}$. 
  Let $\sigma_C$ and $\sigma_{C^\perp}$ be the projections onto $C$ and $C^\perp$, respectively, which satisfy
  \[
    \sigma_C \circ \sigma_{C^\perp} = \sigma_{C^\perp} \circ \sigma_C = 0, \quad \sigma_C(\vv) + \sigma_{C^\perp}(\vv) = \vv.
  \]
  Then, there exist $\xx_C \in \ZZ_k^m$ and $\xx_{C^\perp} \in \ZZ_k^{n-m}$ such that $\sigma_C(\vv) = \xx_C G$ and $\sigma_{C^\perp}(\vv) = \xx_{C^\perp} H$. This implies
  {\small 
    \begin{align*}
      \vv &= \sigma_C(\vv) + \sigma_{C^\perp}(\vv) \\
      &= \xx_C G + \xx_{C^\perp} H \\
      &= \xx \begin{pmatrix}
        G \\
        H
      \end{pmatrix}
    \end{align*}
  }
  Therefore,
  \[
  \xx = \vv \begin{pmatrix}
      G^\top (GG^\top)^{-1} & H^\top (HH^\top)^{-1}
    \end{pmatrix}
  \]
  and
  \begin{align*}
    \sigma_C(\vv) &= \vv G^\top (GG^\top)^{-1} G,\\
    \sigma_{C^\perp}(\vv) &= \vv H^\top (HH^\top)^{-1} H.
  \end{align*}
  From the above, we have
  \[
      C = \text{Im}(\sigma_C) = \left\{ \vv G^\top (GG^\top)^{-1} G \mid \vv \in \ZZ_k^n \right\}.
  \]
  If $C_1$ and $C_2$, which are free and LCD over $\ZZ_k$, are permutation code equivalent, then there exists $P$ such that $G_1 = G_2 P$.
  Then,
  \begin{align*}
    G_1^\top (G_1 G_1^\top)^{-1} G_1 &= (G_2 P)^\top (G_2 P (G_2 P)^\top)^{-1} G_2 P \\
    &= P^\top G_2^\top (G_2 G_2^\top)^{-1} G_2 P,
  \end{align*}
  and we obtain that $G_1^\top (G_1 G_1^\top)^{-1} G_1$ and $G_2^\top (G_2 G_2^\top)^{-1} G_2$ are graph isomorphisms.
  On the contrary, assume there exists $P$ such that 
  \[
    G_1^\top (G_1 G_1^\top)^{-1} G_1 = P^\top G_2^\top (G_2 G_2^\top)^{-1} G_2 P.
  \]    
  Then, 
  \begin{align*}
    C_1 &= \{\vv G_1^\top (G_1 G_1^\top)^{-1} G_1 \mid \vv \in \ZZ_k^n\} \\
    &= \{\vv P^\top G_2^\top (G_2 G_2^\top)^{-1} G_2 P \mid \vv \in \ZZ_k^n\} \\
    &= \{\vv G_2^\top (G_2 G_2^\top)^{-1} G_2 P \mid \vv \in \ZZ_k^n\} \\
    &= C_2 P
  \end{align*}
  and $C_1$ is permutation code equivalent to $C_2$.
\end{pproof}

Note that the most important part of the proof of \Cref{thm:1a} is that $GG^\top$ is non-singular.
Since codes over any field are always free, we can prove \Cref{thm:1} for codes assuming only that they are LCD.

\subsubsection{The specific case when $k$ is an odd prime power}

In fact, when $k$ is an odd prime power, we can eliminate the property of freeness of the code.
\rev{By the following argument,} if a ring $R$ has the property called FCR (finite chain ring), then the LCD code over $R$ is always free.

\rev{
  \begin{df}[Finite chain ring (FCR)]
    A finite commutative ring with unity is called a finite chain ring (FCR) if, for any two ideals, one is contained in the other.\com{B}{8}
  \end{df}
}

\begin{prop}[\cite{DG2023}]\label{prop:1}
  Let $R$ be an FCR, let $C$ be a code over $R$ with length $n$, and let $G$ be a generator matrix of $C$.
  Then, the following are equivalent:
  \begin{itemize}
    \item[(i)] $C$ is LCD.
    \item[(ii)] $GG^\top$ is non-singular.
  \end{itemize}
\end{prop}

\rev{Note that $GG^\top$ is non-singular if and only if $C$ is free and LCD code.}
Liu-Liu~\cite{LL2015} shows that \Cref{prop:1} holds for specific LCD codes, and Durǧun~\cite{D2020} shows \Cref{prop:1} holds for any LCD code.
It is also known that $\ZZ_k$ is an FCR if $k$ is a prime power,
so in this case, we can obtain \Cref{thm:main}.

\begin{cor}\label{thm:main}
  Let $q$ be an odd prime power, and let $C$ be an LCD code over $\mathbb{Z}_{q}$ with length $n$. 
  For any $i \in \{1, 2\}$, let $O_i$ be an orthonormal matrix of size $n$, and let $L_i = O_i(C + p^a \mathbb{Z}^n)$. 
  Then, LIP for $L_1$ and $L_2$ can be reduced to $\mathbb{Z}$LIP and GI with $2n$ vertices.
\end{cor}

\subsection{The case $k$ is even and not divisible by $4$}\label{ssec:4-2}

Additionally, we consider the case when $k$ is even. In this case, the obstacle to generalization is \Cref{lem:11a}. To resolve this, we modify the signed closure code.
\begin{df}[Extended signed closure]\label{df:esc}
  Let $k=2m$ be an even number, where $m$ is odd, and let $T_m$ be a matrix of size $n \times 3n$ as follows:
  \[
    T_m = \begin{pmatrix}
      T & mI_n
    \end{pmatrix},
  \]
  where $T$ is defined in \Cref{lem:10a}.
  For any code $C$ over $\ZZ_k$ with length $n$, the extended signed closure code of $C$, denoted as $C^{e}$, is a code over $\ZZ_k$ with length $3n$, defined as follows:
  \[
    C^e \coloneq \{ cT_m \mid c \in C \}.
  \]
\end{df}
Extended signed closure code is a code whose first $2n$ rows are equivalent to a signed closure code, and after the $(2n+1)$-th row, all entries are $m$ or $0$. 
When $k$ is even and not divisible by $4$, the SPEP for free and LCD codes over $\ZZ_{k}$ with length $n$ can be reduced to the PEP for free and LCD codes over $\ZZ_{k}$ with length $3n$ using the extended signed closure code. 
Firstly, we show \Cref{lem:10b}, which corresponds to the generalization of \Cref{lem:10a}.

\begin{lem}\label{lem:10b}
  Let $m$ be an odd \rev{integer} and $k=2m$.
  The signed permutation code equivalence of $C_1$ and $C_2$, which are over $\ZZ_k$, is equivalent to the permutation code equivalence of $C_1^e$ and $C_2^e$.
\end{lem}

\begin{proof}
  In this proof, if we denote the column vectors $\{\vv_1, \ldots, \vv_r\}$, we consider this as a multiset and
  \[
    \{\uu_1, \ldots, \uu_r\} = \{\vv_1, \ldots, \vv_r\}
  \]
  denote this as an equivalence of multisets.
  Let $n$ be the length of $C_1$ and $C_2$.
  For any $i \in \{1, 2\}$, let $G_i$ be a generator matrix of $C_i$, and let $G_i^e$ be a generator matrix of $C_i^e$.
  Then, from the definition, $G_i^e = G_i T_m$.
  Firstly, we assume that $C_1$ and $C_2$ are signed permutation code equivalent, i.e., there exists a permutation matrix $P$ and a diagonal matrix $D$ such that $G_1 = G_2 P D$.
  Thus, $G_1^e = G_1 T_m = G_2 P D T_m$.
  Since $m = -m$,
  \[
    D T_m = \begin{pmatrix}
      D T & D m I_n
    \end{pmatrix} = \begin{pmatrix}
      D T & m I_n
    \end{pmatrix}.
  \]
  From \Cref{lem:10a}, since there exists a permutation matrix $P'$ such that $D T = T P'$, there also exists a permutation matrix $P'_{3n}$ such that 
  $P'_{3n} = \begin{pmatrix}
    P' & 0 \\
    0 & I_n
  \end{pmatrix}$ and $D T_m = T_m P'$.
  From the proof of \Cref{lem:10a}, because there exists a permutation matrix $P''$ such that $P T = T P''$,
  $P T_m = \begin{pmatrix}
    P T & P m
  \end{pmatrix} = \begin{pmatrix}
    T P'' & m I_n P
  \end{pmatrix} = T_m \begin{pmatrix}
    P'' & 0 \\
    0 & P
  \end{pmatrix}.$
  Therefore,
  \[
    G_1^e = G_2 PDT_m = G_2 T_m P_{3n}' P'' = G_2^e P''',
  \]
  this implies that $C_1^e$ is permutation code equivalent to $C_2^e$.
  On the contrary, we assume there exists a permutation matrix such that $G_1 T_m = G_2 T_m P$.
  Let
  \[
    G_1 T_m = (\uu_1, \ldots, \uu_{3n}), \quad G_2 T_m = (\vv_1, \ldots, \vv_{3n}),
  \]
  and let $\sigma$ be a permutation corresponding to $P$, such that $\uu_i = \vv_{\sigma(i)}$, and define $U_m$ and $V_m$ as follows:
  \begin{align*}
    U&\coloneq \{\uu_1,\cdots,\uu_{3n}\},\\
    V&\coloneq \{\vv_1,\cdots,\vv_{3n}\},\\
    U_m &\coloneq \{\uu_i \in U \mid \uu_i = m \uu_i\}, \\
    V_m &\coloneq \{\vv_i \in V \mid \vv_i = m \vv_i\}.
  \end{align*}
  Note that all entries of the column vector $\uu_i$ are $0$ or $m$ if and only if $\uu_i \in U_m$, and $V_m$ is the same. 
  Especially, if $i \geq 2n+1$, then $\uu_i \in U_m$ and $\vv_i\in V_m$.

  Since $U=V$, $V_m = U_m$, and $\sigma$ can be represented as the product of permutations $\sigma_m$, which maps $V_m$ to $U_m$, and a permutation $\tau$, which maps $V \setminus V_m$ to $U \setminus U_m$.
  From now on, we show the existence of $P'$ such that $G_1 T_m = G_2 T_m P'$ and $P' = \begin{pmatrix} P_{2n} & 0 \\ 0 & P_n \end{pmatrix}$, where $P_{2n}$ and $P_n$ are permutation matrices of size $2n$ and $n$, respectively.
  Note that $P'$ represents the product of permutations $\sigma_1$ and $\sigma_2$, where $\sigma_1$ is a permutation of the first $2n$ rows and $\sigma_2$ is a permutation of the last $n$ rows.
  Define
  \begin{align*}
      U_{\leq 2n} &\coloneq \{\uu_i \in U_m \mid i \leq 2n\}, \\
      V_{\leq 2n} &\coloneq \{\vv_i \in V_m \mid i \leq 2n\}.
  \end{align*}
  Then, if $U_{\leq 2n}=V_{\leq 2n}$ then such a $P'$ exists.
  Define $V_m^c \coloneq V\setminus V_m$ and $U_m^c \coloneq U\setminus U_m$.
Then, $\tau(V_m^c) = U_m^c$.
From the definition of the extended signed closure code, if $\uu_i \in U_m^c$, then $\uu_{\lceil \frac{i}{2} \rceil + 2n} = m\uu_i$.
Define
\begin{align*}
  U_m^{c+} & \coloneq \left\{ \uu_{\frac{i}{2} + 2n} \mid \uu_i \in U_m^c, \text{ $i$ is even} \right\}, \\
  U_{\leq 2n}^+ & \coloneq \left\{ \uu_{\frac{i}{2} + 2n} \mid \uu_i \in U_{\leq 2n}, \text{ $i$ is even} \right\}, \\
  V_m^{c+} & \coloneq \left\{ \vv_{\frac{i}{2} + 2n} \mid \vv_i \in V_m^c, \text{ $i$ is even} \right\}, \\
  V_{\leq 2n}^+ & \coloneq \left\{ \vv_{\frac{i}{2} + 2n} \mid \vv_i \in V_{\leq 2n}, \text{ $i$ is even} \right\}.
\end{align*}
% Then, $\{\uu_{2n + 1}, \ldots, \uu_{3n}\} = U_m^{c+} \cup U_{\leq 2n}^+ = V_m^{c+} \cup V_{\leq 2n}^+$.
Since $U_m^c=V_m^c$, $U_m^{c+}=V_m^{c+}$.
Because $U_m = U_m^{c+} \cup U_{\leq 2n} \cup U_{\leq 2n}^+$ and $V_m$ is the same, we obtain
\[
  U_{\leq 2n} \cup U_{\leq 2n}^+ = V_{\leq 2n} \cup V_{\leq 2n}^+.
\]
Additionally, since $m = -m = m^2$, we have
\begin{align*}
    \{\uu_{i}\in U_{\leq 2n}\mid \text{$i$ is even}\} &= \{\uu_{i}\in U_{\leq 2n}\mid \text{$i$ is odd}\} \\
    &= \{\uu_{i}\in U_{\leq 2n}^+\},
\end{align*}
which shows that the multiset $\{\uu_{i} \in U_{\leq 2n} \cup U_{\leq 2n}^+\}$ contains all vectors of $\{\uu_{i} \in U_{\leq 2n}, \text{$i$ is odd}\}$ exactly three times.
The same equation holds for $V_{\leq 2n}$.
Therefore,
\[
    U_{\leq 2n} = V_{\leq 2n}.
\]
From the above, there exists a permutation matrix $P' = \begin{pmatrix}
    P_{2n} & 0 \\
    0 & P_n
\end{pmatrix}$ such that
\begin{align*}
    \begin{pmatrix}
        G_1 T & m G_1
    \end{pmatrix} &= G_2 T_m P' \\
    &= \begin{pmatrix}
        G_2 T P_{2n} & m G_2 P_n
    \end{pmatrix}.
\end{align*}
This especially implies $G_1 T = G_2 T P_{2n}$ and $m G_1 = m G_2 P_n$.
The first equation shows a permutation code equivalence of $C_1^\pm$ and $C_2^\pm$, and from \Cref{lem:10a}, $C_1$ and $C_2$ are signed permutation code equivalent.
Therefore, if $C_1^e$ and $C_2^e$ are permutation code equivalent, then $C_1$ and $C_2$ are signed permutation code equivalent.
\end{proof}

\begin{lem}\label{lem:11b}
  Let $m$ be an odd \rev{integer} and $k=2m$.
  A code $C$ over $\ZZ_k$ is free and LCD if and only if $C^e$ is free and LCD.
\end{lem}

\begin{proof}
  The freeness of $C^e$ also holds by the same argument as in \Cref{lem:11a}.
  We now consider the LCD property.
  For any vectors $\xx=(x_1,x_2,\ldots,x_n)$ and $\yy=(y_1,y_2,\ldots,y_n)$ in $C$,
  \begin{align*}
    \xx^e &= (x_1,-x_1,\ldots,x_n,-x_n,mx_1,\cdots, mx_n), \\
    \yy^e &= (y_1,-y_1,\ldots,y_n,-y_n,my_1,\cdots,my_n)  
  \end{align*}
  are also in $C^e$, and 
  \[
    \xx^e\cdot \yy^e = (m^2+2)\xx\cdot\yy.
  \]
  Since $k=2m$ and $m$ is odd, $m^2+2$ is a unit over $\ZZ_k$.
  Therefore, $\xx\cdot\yy=0$ and $\xx^e\cdot \yy^e=0$ are equivalent.
  This implies that $C$ is LCD if and only if $C^e$ is LCD.
\end{proof}

From \Cref{lem:10b} and \Cref{lem:11b}, we can generalize \Cref{thm:maina} for $k$ being an even number and not divisible by $4$.

\begin{thm}\label{thm:mainb}
  Let $k$ be an even number and not divisible by $4$, and let $C$ be a free and LCD code over $\ZZ_{k}$ with length $n$.
  For any $i \in \{1, 2\}$, let $O_i$ be an orthonormal matrix of size $n$ and define $L_i = O_i(C + k\ZZ^n)$.
  Then, the LIP of $L_1$ and $L_2$ can be reduced to $\ZZ$LIP and GI \rev{for graphs with vertex size} $3n$.
\end{thm}

\begin{proof}
  From \Cref{lem:9a}, the LIP for lattices obtained by Construction A from any LCD code over $\ZZ_k$ with length $n$ can be reduced to the $\ZZ$LIP and SPEP for LCD codes over $\ZZ_k$ with length $n$, and from \Cref{lem:10b}, this SPEP can be reduced to the PEP for codes over $\ZZ_k$ with length $3n$. 
  Additionally, from \Cref{lem:11b}, the code obtained by \Cref{lem:10b} preserves the properties of being free and LCD. 
  Finally, from \Cref{thm:1a}, the PEP for free and LCD codes can be reduced to the graph isomorphism problem, yielding \Cref{thm:mainb}.    
\end{proof}

% \Cref{lem:10b}、\Cref{lem:11b}から、$k$が$4$の倍数でない偶数の場合へ\Cref{thm:maina}を一般化できる。

% \begin{thm}\label{thm:mainb}
%   $k$が$4$の倍数でない偶数として、$C$を$\ZZ_{k}$上の長さ$n$の自由かつLCD符号とする。
%   また、任意の$i\in\{1,2\}$に対して$O_i$を$n$次の実数上の直交変換として、$L_i=O_i(C+k\ZZ^n)$とする。
%   この時、$L_1$と$L_2$の格子同型問題は、$\ZZ$LIPと$3n$頂点のグラフ同型問題へと帰着できる。
% \end{thm}

% %\bibliographystyle{ieicetr}% bib style
% %\bibliography{}% your bib database
% \begin{thebibliography}{99}% more than 9 --> 99 / less than 10 --> 9
% \bibitem{}
% \end{thebibliography}

% %\profile{}{}
% %\profile*{}{}% without picture of author's face

% \section{Future work}\label{sec:5}

% % 今後の課題として、今回の手法では考察できなかった$k$が$4$の倍数となる場合についても$\ZZ$LIPといくつかの頂点のグラフ同型問題へと帰着できるのか、という問題が考えられる。
% % また、符号の同型問題について、今回はサイン付き置換同型と呼ばれる問題を考え、それをサイン閉包符号あるいは拡張サイン閉包符号の置換同型問題へと帰着した。
% % 特に、この符号の変換によって符号のLCD性が変化しないのが重要であった。
% % これを少し一般化して、ある環$R$の部分集合を$S$として、
% % $R$上の線形符号$C_1,C_2$が与えられたとき、$C_1=DPC_2$かつ、$D$の全ての成分が$S$の要素となる時を$S$置換同型と定義しよう。
% % 例えば、サイン付き置換同型とは$S=\{1,-1\}$とした時の$S$置換同型となる。
% % この時、$R$上のLCD符号の$S$置換同型を、$C_1,C_2$から作られるLCDな符号の置換同型問題へと帰着できるような$S$はどのような条件を持つだろうか。

\section*{Acknowledgement}
%This work is supported by JSPS Grant-in-Aid for Scientific Research(C) JP22K11912, JST CREST JPMJCR2113, and MEXT Quantum Leap Flagship Program (MEXT Q-LEAP) JPMXS0120319794.

%This work is supported by JSPS Grant-in-Aid for Scientific Research(C) JP22K11912, JST K Program Grant Number JPMJKP24U2, and MEXT Quantum Leap Flagship Program (MEXT Q-LEAP) JPMXS0120319794.

%This work is supported by JSPS Grant-in-Aid for Scientific Research(C) JP22K11912, JST K Program Grant Number JPMJKP24U2, and MEXT Quantum Leap Flagship 
%Program (MEXT Q-LEAP) JPMXS0120319794.

This work is supported by JSPS Grant-in-Aid for Scientific Research(C) JP22K11912, JST K Program Grant Number JPMJKP24U2, and MEXT Quantum Leap Flagship 
Program (MEXT Q-LEAP) JPMXS0120319794.

\end{document}